%% file: paper.tex
\documentclass[sigconf, nonacm]{acmart}
\usepackage{algorithmic}
\usepackage{algorithm}
\usepackage{amsthm}
\usepackage{amsmath}
\theoremstyle{remark}

\setcopyright{none}

\title{A Fast Randomized Algorithm for Finding the Maximal Common Subsequences}

\author{Jin Cao}
\affiliation{\institution{Nokia Bell Labs}}
\email{jin.cao@nokia-bell-labs.com}

\author{Dewei Zhong}
\affiliation{\institution{Rutgers University}}
\email{dewei.zhong@rutgers.edu}

\begin{document}

\begin{abstract}

Finding the common subsequences of $L$ multiple
strings has many applications in the area of bioinformatics, computational linguistics, and information retrieval. A well known result states 
that finding a Longest Common Subsequence (LCS) for $L$ strings is NP-hard, e.g.,  the computational complexity is exponential in $L$.   In this paper, we develop a randomized algorithm, referred to as {\em Random-MCS}, for finding a random instance of Maximal Common Subsequence ($MCS$) of multiple strings. A common subsequence is {\em maximal} if inserting any character into the subsequence no longer yields a common subsequence. A special case of MCS is LCS where the length is the longest. 
We show the complexity of our algorithm is linear in $L$, and therefore is suitable for large $L$.
Furthermore, we study the occurrence probability for a single instance of MCS, and demonstrate via both theoretical and experimental studies that the longest subsequence from multiple runs of {\em Random-MCS} often yields a solution to $LCS$. 
\end{abstract}

\maketitle

\input{intro}

\input{background}

\input{algorithm}
\input{experiment}

\input{conclusion}
\input{append}
\bibliographystyle{ACM-Reference-Format}
\bibliography{mcs}
\end{document}

%% file: intro.tex
\section{Introduction}
Data discovery and pre-processing in many data science projects often require laborious efforts and creativity from the data scientist. Developing methods that can automatically generate insights from raw data is an important topic in automated machine learning \cite{feurer2015efficient} in order to eliminate human bottleneck and make machine learning available to non-experts. As string or text is a common form of data representation, comparing strings so that information regarding to what is common and what is unique among the strings can be extracted and summarized is an important pre-processing task. 

A subsequence of a string $S$ is a character sequence that can be derived from $S$ by deleting some characters without changing the order of the remaining characters. 
Consider the case of $L$ strings where $L$ is large.
A common subsequence of $L$ strings can be thought of as a common pattern  shared by all strings.
Unlike substrings, subsequences are not required to occupy consecutive positions within the original strings. 

For string comparison, we consider two types of common subsequences of the $L$ strings. 
The Longest Common Subsequence (LCS) is a subsequence 
common to all the $L$ strings that has a maximal length. The Maximal Common Subsequence (MCS) is defined as {\em maximal} if and only if inserting any character into the subsequence can no longer yield a common subsequence.  By definition, a LCS is a MCS with the maximal length. Furthermore, there may exist many 
MCSs of different lengths, and many LCSs of the same maximal length.  For example, for the given two strings  $'fabecd'$ and $'acdef'$, 
the set of MCSs are $\{f, acd, ae\}$ where $acd$ is the LCS.


Finding LCS for multiple strings has important applications in many areas, including 
bioinformatics, computational linguistics, and information retrieval \cite{attwood1994fingerprinting, bourque2002genome,sorokin2016using}. 
The problem is, however, NP-hard  \cite{maier1978complexity} as
the number of strings $L$ becomes large. 
Much of the literature addresses the simple case of two or three strings \cite{hirschberg1975linear,masek1980faster,hakata1992algorithms}. Several methods have been proposed to improve the computation efficiency for the general case of $L$ strings, either by using parallelization \cite{chen2006fast,wang2010fast,korkin2008efficient} or assuming a special string structure  \cite{hakata1998algorithms}. Reviews of various methods can be found in \cite{bergroth2000survey, kawade2017}.


In this paper, we attack the problem of string comparison from the angle of MCS instead of LCS. The problem of finding MCS is much less studied compared to LCS.
All methods from the existing literature only consider the case of two strings.
For example, methods are presented by \cite{hirschberg1975linear} to find MCS and constrained MCS. 
A dynamic programming approach is presented in 
\cite{Fraser:1995:AAS:642129.642130} to find the shortest MCS.
More recently, \cite{sakai2019maximal} proposes an computationally efficient way to find a MCS but his method can only find one MCS. 

We develop a fast randomized algorithm to find MCS solutions of $L$ strings and show the computational complexity is linear in $L$, thus much more amenable
for the analysis of a large number of strings 
than algorithms developed for LCS. Furthermore, as each run of our algorithm returns a random MCS and LCS is the longest MCS, we can run our algorithms multiple times and then take the longest MCS from the returned solutions to approximate LCS. We study this both theoretically and empirically. 
Our main contributions are summarized as follows:
\begin{itemize}
\item  We develop a randomized algorithm, referred to as $RandomMCS$, for finding a random $MCS$ solution of multiple strings.
\item We extend an existing algorithm for finding $MCS$ of two strings \cite{sakai2019maximal}  to the case of $L$ strings. 
\item
For a set of $L$ strings with common length $n$,
we show the computational complexity of our  $RandomMCS$ algorithm is
$O(n^3 L)$ and our extension to the algorithm in \cite{sakai2019maximal}, is $O(nL \log n )$, both are linear in the number of strings $L$.
\item We carry out simulation studies to understand the performance of our proposed approach. 
\item We analyze the occurrence probability of a MCS solution returned from $RandomMCS$.
\item
We demonstrate via both theoretical analysis and experimental studies that the longest subsequence from multiple runs of our algorithm often yields a $LCS$. 
\end{itemize}

The rest of the paper is organized as follows.
In Section 2, we present the relevant background for our work. In Section 3, we propose our method and illustrate it using a toy example. In Section 4, we analyze the occurrence probability for a specific MCS and show the computational complexity of our algorithm is linear in the number of strings $L$.
We carry out simulations to understand the performance of our algorithm empirically and present an application of our work to Automated Machine Learning (AutoML) in Section 6. 
We conclude and discuss future work in Section 7. 


%% file: background.tex
\section{Background}\label{sec:background}
In this section, we shall first formally define 
Longest Common Subsequence (LCS) and 
Maximal Common Subsequence (MCS) for $L$ strings. Then we discuss previous work on finding LCS and MCS. 
We shall introduce the following notations used throughout the paper. We denote the empty string by $''$ and denote the empty set by $\emptyset$. To make presentation clear, we put quote $''$ around single characters to differentiate them from variables but sometimes omit the $''$ for strings with multiple characters.
We use calligraphic letters to indicate sets, i.e., $\mathcal{A}, \mathcal{M}$, etc. Throughout the paper, strings are represented using upper case letters. We use $\oplus$ to represent string join, and  reserve the letter $L$ to indicate the number of strings in consideration.

\subsection{Definitions}
In the following, we are given a set $\mathcal{A}$ of $L$ strings:  $\{A_1, A_2, \ldots, A_L\}$, where each $A_l$ is a string with $n_l$ characters represented by
$A_l=a_{1l} a_{2l} \ldots a_{n_l,l}$. 

\begin{definition}
A sequence of characters $C$ is a common subsequence for (strings in) $\mathcal{A}$, if $C$ is contained in each $A_l, l=1,2,\ldots,L$ in the same character order. 
\end{definition}
To avoid confusion, we differentiate a subsequence from a {\em substring} where a substring a consecutive block of characters from a string. For a subsequence, we often
concatenate its characters and use a string to represent it.

\begin{definition}
Define
$LCS(\mathcal{A})$ as the longest common subsequence contained in each string $A_l$ in $\mathcal{A}, l=1,\ldots,L$.
\end{definition}
\begin{definition}
Define
$MCS(\mathcal{A})$ as a subsequence contained in each string $A_l$ in $\mathcal{A}$ with the property such that an addition of any character to $MCS(\mathcal{A})$ no longer yields a common subsequence for $\mathcal{A}$.
\end{definition}
\textbf{Example.}
The solution set of MCS for $\mathcal{A}=\{TEGAP, GAEPR\}$ is $\{GAP, EP\}$.
Out of these two solutions, $'GAP'$ is the LCS.

\subsection{Algorithms for Finding LCS and MCS}
Dynamic programming is a common technique used for finding LCS. For example, consider the LCS of two strings of length $n$, $X = x_1x_2...x_n$ and  $Y=y_1y_2...y_n$. If $x_n = y_n$, then $LCS(X,Y) = LCS(X_{n-1},Y_{n-1})$ $\oplus x_n$. If $x_n \neq y_n$, then $LCS(X,Y) = \text{max}(LCS(X_{n-1},$ $Y),LCS(X,$ $Y_{n-1}))$ where $X_{n-1}$ and $Y_{n-1}$ represent the previous $n-1$ elements of $X$ and $Y$ respectively. 
It can be shown the complexity of using dynamic programming for finding LCS is $O(n^2)$. For the general case of $L$ strings, the extension of the dynamic programming algorithm will have a time complexity of $O(n^L)$, 
which implies the problem is NP-hard \cite{maier1978complexity}.
An algorithm of a running time of $O((r+n)\log n)$ is proposed by \cite{hunt1977fast} 
where $r$ is the total number of ordered pairs of positions at which the two sequences match. In the worst case $r$ can be  $O(n^2)$.

There are several proposed methods for finding MCS.
It has been shown by \cite{Fraser:1995:AAS:642129.642130} the problem of finding all shortest MCSs for $L$ strings is NP-hard for large $L$. 
All proposed algorithms focus only on two strings and no computationally effective methods have been proposed in the general case of $L$ strings.  Our algorithm targets the general case.

%% file: algorithm.tex
\section{Algorithms to Find Multiple MCSs of $L$ Strings}

\subsection{Intuition}
Our algorithm is inspired by Lemma 2 from \cite{sakai2019maximal} which states a necessary and sufficient condition for a subsequence $W$ being maximal for two strings. We shall extend the lemma to the case of $L$ strings. 
In the following,  we denote the set of $L$ strings of interest by $\mathcal{A}=\{A_1,\ldots,A_L\}$.
\begin{definition}
For a string $A$, define 
$|A|$ as the number of characters in $A$. For each $k=1,\ldots,|A|$, define
$A(0,k]$ as the prefix of $A$ starting from position $1$ to $k$. Define $A(k,|A|]$ as the suffix of $A$ starting from position $(k+1)$ to $|A|$. Define $A(0,k]=''$ for $k=0$ and $A(k,|A|]=''$ for $k=|A|$ where $''$ is the empty string.
\end{definition}

\begin{definition}
Let $W$ be a subsequence contained in string $A$, then for any $k=0,\ldots,|W|$, define 
$Middle(A, W, k)$ as the remaining  substring obtained from $A$ by deleting both the shortest prefix containing $W(0,k]$ and the shortest suffix containing $W(k,|W|]$.
\label{def:middle}
\end{definition}
\textbf{Example.}
The following gives a simple example of this function. $Middle('T$ $EGAP', 'E', k=0)$ is $'T'$ since when 
 $W='E'$ and $k=0$, 
the shortest prefix in $'TEGAP'$ containing $W(0,k]=''$
is $''$, 
and the shortest suffix containing $W(k,|W|)='E'$ is $'EGAP'$ (this example is also shown in the first line in Cell 3 of Figure 2).


\begin{theorem}
For any common subsequence $W$ of
$\mathcal{A}$, $W$ is maximal if and only if
for any $0\leq k \leq |W|$, the set of $L$  substrings 
$Middle(A_l, W, k)$, derived from $A_l$, $l=1,\ldots, L$, are disjoint (i.e. do not share any common characters).
\label{lem:mcs1}
\end{theorem}
\begin{proof}
If $W$ is maximal, then for each $k=0,1,\ldots,|W|$, the $L$ substrings  $Middle(A_l,W,k)$, 
derived from $A_l, l=1,\ldots, L$,  have to be disjoint. This is because if this is not true, then there exisits a common character $c$ shared by the $L$ substrings $Middle(A_l,W,k), l=1,\ldots, L$. Therefore, by (string) joining $W(0,k]$, $c$, and $W(k,|W|)$, we can construct a longer common subsequence that contains $W$ which contradicts the condition that $W$ is maximal. The converse is true since it validates the condition of $W$ being maximal. 
\end{proof}

The contra-positive of the above lemma can be stated as follows. 

\begin{theorem}
For any common subsequence $W$ of
$\mathcal{A}$,  $W$ is not maximal if and only if there exist $k, 0\leq k\leq |W|$ such that the set of $L$ substrings, $Middle(A_l, W, k)$,
derived from $A_l$, 
$l=1,\ldots, L$
share at least one common character.
\label{lem:mcs2}
\end{theorem}

Theorems \ref{lem:mcs1} and \ref{lem:mcs2} are in fact the basis of our algorithm since it can be used  to {\em constructively} obtain a MCS.
Suppose we start $W$ as the empty set, according to Theorem \ref{lem:mcs2}, if $W$ is not maximal, then we can find a 
character that is common to the set of $L$ strings $\mathcal{A}$ to add to $W$. This step can be performed iteratively until $W$ become maximal, i.e., the set of $L$ substrings, $Middle(A_l,W,k)$,  each from $A_l$,  becomes disjoint so that we can no longer insert characters to $W$.  
To obtain many instances of MCSs, we randomize the character insertion to $W$, which is the essence of our algorithm.


\subsection{RandomMCS Algorithm}
To formally present out algorithm, we first need to define some supporting functions. 
\begin{definition}
\label{def:commonchar}
Define {\em commonChar}($\mathcal{A}$) as the function that returns a set of common characters shared by each string in a given string set $\mathcal{A}$.
\end{definition}

\textbf{Example.} Suppose $\mathcal{A}=\{TEGAP, GAEPR\}$, the function will return a set of 4 characters $\{E, G, A, P\}$ as they are all shared characters for the two strings. Suppose $\mathcal{A}=\{abccde,gfchca,dfcca\}$, then the function will return the set $\{a,c\}$.
However, in this case, the character $'c'$ appears at least two times in every string. This frequency information can be used in our algorithm when we randomly select a character from the common set so that
the high-frequency characters are more likely to be selected.

\begin{definition}
Given a set of $L$ Strings $\mathcal{A}=\{A_1,A_2,\ldots,A_L\}$ and 
a common subsequence $W$, define the function $BreakPoints(\mathcal{A}, k)$ that returns the set of location indices $k$ to be inserted in $W$
 so that the new subsequence is still common to all strings in $\mathcal{A}$. That is, the updated common subsequence is the string join of $W(0,k], c, W(k,|W|]$. 
\end{definition}

A pseudo code implemention of the function $BreakPoints(\mathcal{A}, k)$ is shown as follows.
\begin{algorithm}
\caption{Function $BreakPoints$}
\label{alg:segPos}
\begin{flushleft}
\textbf{Input}: A set of $L$ strings $\mathcal{A}$ and a common subsequence $W$ \leavevmode \\
\textbf{Output}: {The list of indices in $W$ where new characters can be potentially inserted to create an updated common subsequence.}
\end{flushleft}
\begin{algorithmic}[1]
\STATE \hspace{0in}{\em position}  $\leftarrow \emptyset$
\STATE \textbf{for} $k$ in $0:|W|$
    \STATE \hspace{0.15in}\textbf{for} $l$ in $1:L$, $m_l \doteq Middle(A_l, W, k)$
\STATE \hspace{0.15in}\textbf{if} $commonChar(\{m_1,\ldots,m_L\})\neq \emptyset$ 
\STATE \hspace{0.3in} $position \leftarrow position \cup \{k\}$
\STATE \hspace{0in} \textbf{return} \em{position}
\end{algorithmic}
\end{algorithm}

\textbf{Example.} The following gives examples of this function.
For the given $\mathcal{A}=\{TEGAP,GAEPR\}$ and a subsequence $W='A'$, $BreakPoints($ $\mathcal{A}, W)$ will return the set $\{0,1\}$. This is because when $k=0$, according Definition~\ref{def:middle},
$Middle('TEGAP', 'A', 0)='TEG'$ and $Middle('GAEPR',$ $'A', 0)='G'$. Hence, since there is a common character $'G'$ shared by $'TEG'$ and $'G'$, the evaluation of existence of common characters in line 4 of Algorithm~\ref{alg:segPos} will succeed. Likewise, when $k=1$, $Middle('TEGAP', 'A', 1)='P'$ and $Middle('GAEPR', 'A', 1)='EPR'$, sharing a common character $'P'$. Therefore $BreakPoints($ $\mathcal{A}, W)$ will return the set $\{0,1\}$. 

On the contrary, for $\mathcal{A}=\{TEGAP,GAEPR\}$ and
$W='GAP'$
$BreakPoints(\mathcal{A}, W)$ will return an empty set. This is because for each $k=0,1,2$, $Middle('TEGAP', 'GAP', k)$ and $Middle('GAEPR',$ $'GAP', k)$ do not share any common characters.


Algorithm~\ref{alg:mcs} presents the pseudo-code of our algorithm for finding a random solution of MCS. The function is written in a recursive fashion and has an optional starting value of $W$
which we shall explain further in Section \ref{subsec:constraint-mcs}. The termination condition of the algorithm is expressed in line
2 which validates $W$ as a MCS by
Theorem~\ref{lem:mcs1}.
Line 3-7 applies Theorem  \ref{lem:mcs2} (which states the contrapositive of Theorem \ref{lem:mcs1}) to constructively search for the possible common characters to update a previous common subsequence $W$. In line 5, when we randomly select a character from the common set, we can utilize the minimum frequency discussed in the example following Definition~\ref{def:commonchar}
as the optional weights. We have found via simulation studies in Section~\ref{sec:simulation} that this performs better for finding the long MCSs.

\begin{algorithm}
\caption{A randomized algorithm, $RandomMCS$, to find a single MCS for $L$ strings $\mathcal{A} = \{A_1,\ldots,A_L\}$}.
\label{alg:mcs}
\begin{flushleft}
\textbf{Input}:  A set of strings $\mathcal{A}$ \leavevmode \\
\textbf{Optional Input}: An initial starting value of $W$ with default $W=\emptyset$\leavevmode\\ 
\textbf{Output}: {A random MCS $M$ of $\mathcal{A}$}
\end{flushleft}
\begin{algorithmic}[1]
\STATE $position \doteq BreakPoints(\mathcal{A}, W)$\leavevmode\\
\STATE \textbf{if} $position=\emptyset$ \textbf{return} $W$
\STATE\textbf{else} $k \doteq$ a random element (index value) from the set $position$
\STATE \hspace*{0.2in} $\mathcal{A}' \doteq \{Middle(A_l,W,k), l=1,\ldots, L\}$
\STATE \hspace*{0.2in} $c\doteq$ a random character from the set $commonChar(\mathcal{A}')$ with or without optional frequency weighting
\STATE \hspace*{0.2in}{$W\leftarrow$  $W(0,k] \oplus c \oplus W(k,|W|)$} \leavevmode \\
\STATE \hspace*{0.2in}\textbf{return} $RandomLCS(\mathcal{A},W)$
\end{algorithmic}
\end{algorithm}

\begin{figure*}[htpb]
\centering
\includegraphics[width=7in]{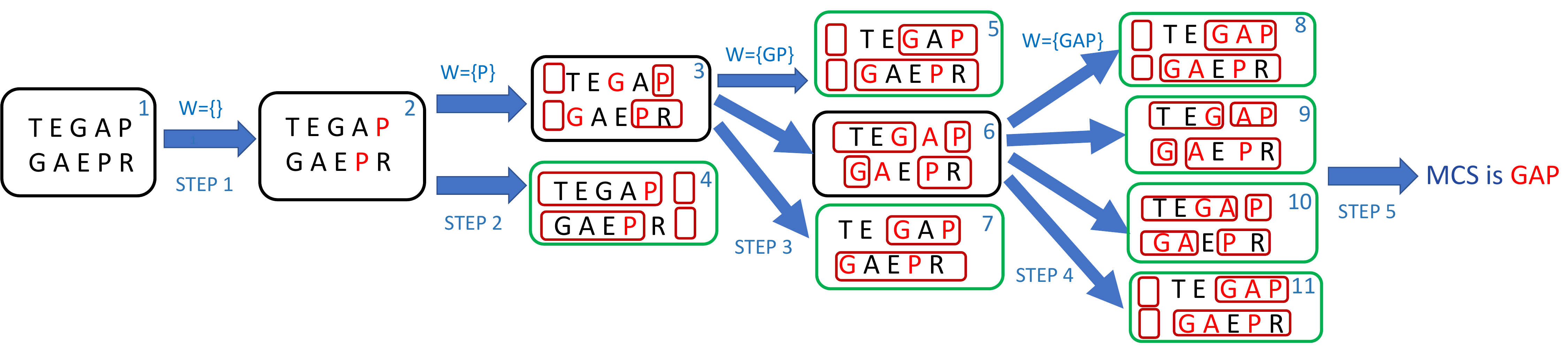}
\caption{Illustration of $RandomMCS$ algorithm for the case of two strings in a run producing $GAP$ as the MCS}
\label{fig:example1}
\end{figure*}

\begin{center}
\begin{figure*}[htpb]
\includegraphics[width=7in]{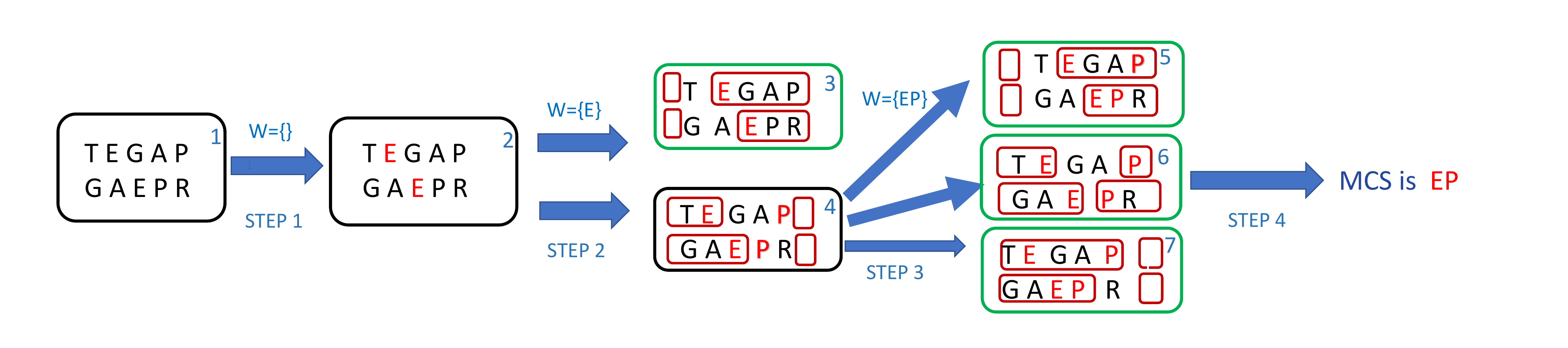}
\caption{Illustration of $RandomMCS$ algorithm for the case of two strings in a different run producing $EP$ as the MCS}
\label{fig:example2}
\end{figure*}
\end{center}

\subsection{A Toy Example}
We shall illustrate our $RandomLCS$ algorithm for finding a random MCS solution using a toy example consisting of two simple strings:$\{TEGAP, GAEPR\}$. We show two runs of the algorithm with different MCS solution output in Figure 1 and 2 respectively. The solutions are different due to the inherent
randomness in the algorithm design. 

Each figure consists of cells that show a certain state of the algorithm through iterations, linked by arrows illustrating the state progression. 
To make the presentation clear, we label each cell with an index value shown in the upper right corner of the cell. Characters in red within each cell represent the current value of the common subsequence $W$ which will be updated through the progression to produce a final MCS solution. The small red frames around the characters indicate the prefix and suffix to be eliminated when 
computing $Middle(A,W,k)$ for a certain $k$ value
(see Definition~\ref{def:middle}), i.e. $Middle(A,W,k)$ is the remaining characters excluding the characters in the red frames.
The outgoing branches from a cell represent the
candidate indices $k=0,\ldots,|W|$ of current common subsequence $W$, in an attempt to update $W$ by
inserting new characters (line 2 of Algorithm~\ref{alg:segPos}). A branch will expire if condition in line 4 of Algorithm~\ref{alg:segPos} is not satisfied, that is, no common characters are found to perform the update.



In Figure 1, we want to find the MCS for the list \{$TEGAP,GAEPR$\} shown in Cell 1. Notice that the two strings share 4 common characters:$'E','G','A','P'$.  Initialize $W=''$. 
Next in Step 1, we choose one of the four characters $'P'$ as the first character to be inserted in $W$, 
and update $W='P'$. We now move to Cell 2 where \{$P$\} is marked red. Since the length of $|W|=|'P'|=1$, we have two places to insert characters in $W$, $k=0,1$, corresponding to the two branches from Cell 2, resulting Cell 3 and Cell 4, respectively. 

We will discuss Cell 4 first, which corresponds to the case of $W='P'$ and $k=1$. In this case, since $Middle('TEGAP', W, k)=''$ and $Middle('GAEPR', W, k)=$ $'R'$ do not share any common characters, the cell expires (recall the red frames indicate the prefix and suffix to be removed for calculating $Middle(A,W,k)$). On the other hand, in Cell 3 where $W='P'$ and $k=0$, $Middle('TEGAP', 'P', k=0)=$ $'TEGA'$ and $Middle('GAEPR', 'P',$ $k=0)='GAE'$, sharing both $'E'$ and $'G'$ as common characters. The progression continues and we select the character $'G'$ to be inserted in $W$ at position 0, resulting an updated $W='GP'$. In summary, at the end of Step 2, $BreakPoint(\mathcal{A}, 'P')={0}$ and character $'G'$ is randomly selected to obtain an updated common sequence $W='GP'$.

\par

By the same token, from Cell 3, 
since $W='GP'$, there are three outgoing branches for $k=0,1,2$ respectively. Similar analysis shows $BreakPoint(\mathcal{A}, W)=\{1\}$ which implies Cell 5 and 7 will expire, and only Cell 6 will continue to the next step. In Cell 6, character $'A'$ is selected so the updated common subsequence is now $W='GAP'$. In Step 4,  $BreakPoint(\mathcal{A}, W)$ returns an empty set which marks the end of the algorithm, resulting $'GAP'$ as the returned MCS output.


Figure~\ref{fig:example2} shows a different realization of our algorithm for the same string pair.
The first difference from Figure 1 occurs in
Cell 2 where the character $'E'$ s added to the common subsequence $W$ instead of $'P'$. Next in Step 2, characer $'P'$ is selected to result a final MCS output of $'EP'$.

\subsection{Constrained MCS}
\label{subsec:constraint-mcs}
A constrained MCS is a MCS that must include a predefined subsequence $W_0$. It is in fact straightforward to modify our algorithm to obtain constrained MCS, simply by using $W_0$ as the starting value (the optional input in the pseudo-code shown in Algorithm~\ref{alg:mcs}). This is due to the nature of our algorithm design as it incrementally inserts a new character to update an existing common subsequence until it becomes maximal.
For instance, consider the constrained MCS problem for the input string set $\mathcal{A}=\{TEGAP,GAEPR\}$ 
that has to contain $'GP'$. Using $'GP'$ as the optional input in Algorithm~\ref{alg:mcs}, the derivation process is identical to Figure \ref{fig:example1} when Cell 3 is used as the starting point. Branches from Cell 3 will finally lead to  $'GAP'$ as the MCS output.

We comment here that \cite{sakai2019maximal}
presented an algorithm for the constrained MCS problem  in the case of two strings. However, 
the modification from the base algorithm used to derive a single  MCS solution is significant.

\section{Analysis of $RandomMCS$ Algorithm}
In this section, we analyze the performance of $RandomMCS$ algorithm. First, for each $MCS$ solution, we study the probability of the solution being returned from one run of the algorithm. We analyze LCS as a special instance of MCS and discuss the probability of a LCS being returned from the algorithm. Next, we analyze the computational complexity of our algorithm and compare it to previous approaches. 
As previous approaches for finding MCS only applies to two strings, we also propose an extension of a previous solution to the case of multiple strings. 

\subsection{Probability Analysis}
As a set of strings may have many MCSs, we denote the set of MCSs as $\mathcal{M}$. Note that one run of our RandomMCS algorithm will yields exactly one random MCS $M$ from the set $\mathcal{M}$, a natural question
to ask is what is the probability value of $M$ being returned from a single run. 

\begin{theorem}
For a given MCS $M$ in the solution set $\mathcal{M}$, the probability of $M$ being returned as the solution from $RandomMCS$ depends only on $M$ and the solution set $\mathcal{M}$.
For a given subsequence $W_0$, let $\mathcal{M}(W_0)$ be the set of MCSs that contains $W_0$ as a subsequence. Then for any $M\in \mathcal{M}(W_0)$, the probability that $M$ being returned as the solution from constrained $RandomMCS$ depends only on $M$ and $\mathcal{M}(W_0)$. This implies that the probability is conditionally independent of the set of $L$ strings, $\mathcal{A}$.
\end{theorem}
\begin{proof}
Notice that each character insertion to an existing common subsequence $W$ (line 3-6 of Algorithm 2) is carried out by two random selections. The first is the choice of a breakpoint position $k$ (line 3) and the second is the choice of a common character $c$ (line 5). Both random selections depend only on the current $W$ and the set of MCSs. Therefore, the random selection is conditionally independent of the original set of strings given $\mathcal{M}$. Hence the result.
\end{proof}

{\textbf{Example.}}
We evaluate the occurrence probability of each MCS being returned from one run of $RandomMCS$ using examples in Figure~\ref{fig:example1} and \ref{fig:example2} where the set of strings under consideration are $\{TEGAP,GAEPR\}$. The solution set of MCS is $\{GAP, EP\}$. Starting with an empty string $W$, notice that we have 4 common characters \{$E,G,A,P$\} in the beginning and all of them share the same probability 1/4 to be selected.
If the first selected character is $'G'$ or $'A'$, the final MCS produced must be $'GAP'$. Likewise, the MCS is $'EP'$ when the first character selected is $'E'$. But when the first character is $'P'$, the returned solution depends on the second selected character. In this case, the choice of first two characters are \{$GP,AP,EP$\} and all of them have the same occurrence probability of 1/3. In total, the probability of $'GAP'$ is $1/4+1/4+1/4 \cdot 2/3=2/3$ and that of $'EP'$ is $1/4+1/4 \cdot 1/3=1/3$.
In this case, we can see our algorithm favors the longer MCS (the LCS) since it has a higher probability. 

\begin{theorem}
\label{thm:mcsprob}
Let $C$ be an upper bound of the number of unique 
common characters  for string set $\mathcal{A}$, i.e.,
$|CommonChars(\mathcal{A})|\leq C$. If $M \in \mathcal{M}$ is a MCS that has a distinguishing subsequence with length bounded by $D$ and  the character 
is selected uniformly random in line 5 of Algorithm~\ref{alg:mcs}, then
it is easy to show that
\[
P(M) \geq C^{-D}. 
\]
This implies the occurrence probability of $M$ is bounded below.
\end{theorem}
\begin{proof}
Let $S$ be a distinguishing subsequence for a MCS $M$
with length bounded by $D$, which implies that 
 $M$ is the only MCS containing $S$. 
 Therefore, if $S$ is selected as the common subsequence after at most $|D|$ character insertions to the initial empty string, then $M$ would be returned as the output MCS from the $RandomMCS$ algorithm. It is now clear that the probability of returning $M$ is bounded by the probability of 
selecting $S$ as the common subsequence after at most $|D|$ character insertions into the initial empty string. If the characters are chosen uniformly, then this probability is bounded by $C^{-D}$.  
\end{proof}

For our toy example where the string set is $\{TEGAP,GAEPR\}$ and the solution set of MCS is $\{GAP, EP\}$. Notice that the number of unique common characters is $C=4$. In addition, either $'G'$ or $'A'$ is a distinguishing subseqeunce for MCS $'GAP'$, therefore, the probability of $GAP$ is bounded by $2C^{-D}=2\cdot 4^{-1}$ = 1/2. Obviously this is a loose lower bound since we have shown before that the actual probability is $2/3$.

For a specific MCS $M$, if the occurrence probability of $M$ is bounded below by a value  $p$, then with enough independent runs of $RandomMCS$ algorithm we can recover $M$ with a high probability.
In fact, for an arbitrarily small $\epsilon$, if we set 
$$T=\left\lceil\frac{\log \epsilon}{\log (1-p)}\right\rceil, $$
then 
\[
P(M \mbox{\ does not appear in $T$ runs}) \leq \epsilon. 
\]
As LCS is a special case of MCS, this implies that if the condition of Theorem~\ref{thm:mcsprob} holds for a LCS, then we can recover the LCS with high probability with enough runs of the algorithm. Hand-waving arguments suggest that our algorithm favors longer MCS  as  
it will likely to contain more characters and more positions (from Algorithm~\ref{alg:segPos}) to be selected to $W$. 
In fact in the extreme case where a MCS contains is formed by multiple occurrences of a single distinct character, it will not be returned unless the character is selected at the first time. In Section~\ref{sec:simulation}, we shall study empirically the occurrence probability of a MCS and correlate that with its length.


\subsection{Complexity Analysis}

\begin{theorem}
For a set of $L$ strings
$\mathcal{A} = \{A_1,A_2,\ldots,A_L\}$, 
let $n_l$ be the string length of $A_l$, $l=1,\ldots,L$. 
Define $n_0 = \min(n_1,n_2,...,n_L)$ as the minimum string length, then the time complexity for one run of Algorithm $RandomMCS$ (Algorithm~\ref{alg:mcs})
to find a MCS solution for $\mathcal{A}$
is $O(n_0^2 \sum_{i=1}^{L} n_i)$. Therefore, when all strings are of equal length $n$, the time complexity is
$O(n^3L)$.
\end{theorem}

\begin{proof}
It is easy to show that the computational complexities of $BreakPoint$ (Algorithm \ref{alg:segPos}) and 
$commonChar$ (Definition \ref{def:commonchar}) are $O(n_0 \sum_{i=1}^L n_i)$ and $O(n_0L)$, respectively. The algorithm $RandomMCS$ may replicate $BreakPoint$ evaluations at most $n_0$ times. Hence the result.

\end{proof}
The above theorem states that the time complexity of our algorithm is linear in the number of strings $L$, as opposed to exponential in $L$ for algorithms to find LCS. It is therefore much more ameanable for the case of large number of strings.

\subsection{Comparison with Previous Approaches}

We compare our approach to previous approaches for finding MCSs.

\subsubsection{Extension of MCS Calculation}
All previous approaches for finding MCSs are developed for the case of two strings \cite{hirschberg1975linear,Fraser:1995:AAS:642129.642130,sakai2019maximal}. The recent algorithm in \cite{sakai2019maximal} can be extended to the case of multiple strings in the following manner. The original algorithm maintains a sequence of 
index pairs that tracks the matches between two strings. We extend their technique and maintain a sequence of $L$-tuple indices that tracks the matches between the $L$ strings. These $L$-tuple indices break the original strings into blocks, where additions to the sequence of $L$-tuples are searched within the matched blocks.
We present the pseudo code in the appendix. 




\subsubsection{Computational Complexity Comparisons}
For two strings with equal length $n$, 
\cite{sakai2019maximal} has the highest efficiency among all proposed algorithms for finding a MCS for two strings. The complexity is $O(n\text{log}(n))$. Our extension to the case of $L$ strings (see appendix) also enjoys the highest efficiency with a complexity $O(Ln\log n)$. However, 
since the  algorithm maintains a certain order when traversing the strings,
it can only find  one MCSs (or two MCSs if we reverse the order of strings), which may not be desirable when there are multiple MCSs.
\cite{Fraser:1995:AAS:642129.642130} focuses on finding MCS first, and obtain all MCSs and the LCS for two strings with length $m\ \text{and}\ n$ with a complexity $O(mn(m+n))$. \cite{hirschberg1975linear} developed
an algorithm for the constrained LCS for two strings with lengths $m$ and $n$ and a complexity $O(mn)$. \cite{hunt1977fast} provides an algorithm to  compute the LCS for 2 strings in the complexity $O(n\log(n))$, but it is only for the special best case scenario with a short LCS.
The following table summarizes the computational complexities of these different methods.

\begin{table}[ht]
\label{tab:mcs.complexity}
\begin{center}
\begin{tabular}{c|c|c}
 Algorithm & Target & Complexity  \\\hline 
 $RandomMCS$ & MCS, $L$ strings & $Ln^3$  \\
 Our extension to Sakai (2019) 
 & MCS, L strings & $O(Ln\log n)$ \\
 Sakai (2019) \cite{sakai2019maximal}
 & MCS, 2 strings & $O(n\log n)$ \\
 Fraser \& Irving(1995) \cite{Fraser:1995:AAS:642129.642130} & MCSs, 2 strings  & $mn(m+n)$\\
 Hirschberg(1975) \cite{hirschberg1975linear} & CLCS, 2 strings  & $mn$\\
 Hunt \& Szymanski(1977)\cite{hunt1977fast} & LCS, 2 strings & $\geq O(n\log(n))$ \\\hline
\end{tabular}
\caption{Comparison of Computational Complexity (CLCS stands for constrained LCS)}
\end{center}
\end{table}

%% file: experiment.tex
\section{Simulation Study}
\label{sec:simulation}
In this section, we perform simulation studies to understand the performance of our $RandomMCS$ algorithm. 
First, we would like to  understand empirically if the longest MCS from multiple runs of $RandomMCS$ 
would yield a solution to LCS.
Second, we study empirically the computational complexity of our algorithm. 

\subsection{Less than 5 Strings}
In this setting, our simulations are run with the number strings varies from 2 to 4, with string lengths ranging from $20$ to $50$. We also vary the alphabet size from 5 to 100. For this experiment, we use the basic dynamic programming method to compute LCS, and run our $RandomMCS$ algorithm 1000 times to select the longest one and compare the result with the real LCS.  The reason that we stop at 4 strings is due to the explosion of the computational time used for finding LCS using dynamic programming when the number of strings exceeds 5.

To simplify the evaluation, strings are generated using random characters from the alphabet. We also consider two kinds of randomization when implementing $RandomMCS$. For the first kind, when we randomly insert a character into a common sequence (line 5 of Algorithm~\ref{alg:mcs}), 
we uniformly choose the character from the common set. For the second kind, we use frequency weighting to select the character 
with a weight that is proportional to the (least) number of times the character appears in each string.  

Some sample results are described as follows. For $L=4$ random strings each with length $n=50$ from an alphabet of size $B=6$, both LCS algorithm and the longest MCS solution from 1000 iterations of $RandomMCS$ yield the same string with length 15. The longest MCS solution took 3sec, and the LCS solution takes 8sec. 
For $L=4, n=50$ and alphabet size $B=50$, longest MCS from our algorithm also yields the same result as the real LCS. In fact, we have not encountered a case where they disagree. Furthermore, the 1000 repetitions are unnecessary for finding LCS using our algorithm as the real LCS tends to have a high occurrence probability being returned (close to 40\%) in many instances. Finally, we do not find significant differences in the performance between the two types of random selection. 


\subsection{A Large Number of Strings}
When the number of strings $L$ gets large, existing algorithms for finding LCS fails to  work well due to the high computational time. We use the following approach to evaluate our algorithm in this instance. Our simulation is designed in such a way that finding the longest common subsequence is challenging. 

Our simulation generates $L=1000$ strings of length 60 in the following manner. First, we generate 4 common subsequences that are contained in each of the 1000 strings: $S_1, S_2, S_3, S_4$ with increasing lengths 3, 6, 9, and 12 respectively, from an alphabet size of 15. 
Next, we insert these subseqeunces into a string of 60 characters in the following way. First we randomly pick 3 indices to situate $S_1$, then we randomly pick 6 indices to situate $S_2$ from the remaining 57 indices, then we randomly pick 9 indices to situate $S_3$ from the remaining 51 indices, and finally we randomly pick 12 indices to situate $S_4$ $S_4$ from the remaining 42 indices. This way all the subseqeunces $S_1,S_2,S_3,S_4$ will be intermingled in each string
which makes the problem of finding LCS challenging.
Notice that the total number of characters in  $S_1,S_2,S_3,S_4$ is 30. 
In the last step of the string generation, we insert 30 random characters into the remaining 30 slots, with an expanded alphabet size of 30 (which includes the original alphabet set of size 15 for $S_1,S_2,S_3,S_4$).

For two random strings with a common length $n$ where
characters are randomly generated from an alphabet,
let the expected length of their LCS be $e$. It has been shown that $\lim_{n\rightarrow \infty} e/n <1$ \cite{chvatal1975longest, kiwi2005expected}. 
Therefore it is easy to conclude that
the expected length of LCS of $L$ such random strings will decrease to 0 exponentially fast with $L$. Since in the last step where we generated 30 completely random characters, with a large $L$, we expect the common subsequence from these 30 random characters will be negligble (or empty). 
Therefore, by design, we expect 
the long subsequences in $S_1,S_2,S_3,S_4$  will remain as MCS
and $S_4$ will be LCS since it is the longest.

The result of our simulation is as follows. With 200 runs of $RandomMCS$, the empirical estimate of the occurrence probabilities for each $S_i,i=1,\dots,4$, is: 0.27 for $S_4$, 0.23 for $S_3$, 0.11 for $S_2$ and a zero probability value for $S_1$. The reason that 
 $S_1$ is no longer a MCS is due to the intermingling of $S_1,S_2,S_3,S_4$ among themselves during the process of situating $S_1,S_2,S_3,S_4$, as the mixing 
 creates spurious common subsequences and $S_1$ is short enough to be absorbed by other MCS solutions. In fact, it is absorbed in one of returned MCS solutions with length 4 (so an extra character was included) and a probability value of 2\%. The intermingling also creates other MCS solutions which accounts for the remaining 38\% of the returned MCS solutions with lengths ranging from 4 to 11. We also varied the alphabet size in the experiment, and found that 
 the intermingling will decrease with larger alphabet size and therefore it would be easier to locate LCS.

To understand the impact of
frequency weighting in the
random character selection (line 5 of Algorithm~\ref{alg:mcs}) 
and the number of characters in the long common sequence 
on the performance of $RandomMCS$ algorithm, we perform the following 2 by 2 experiments. We have two settings for the weights: uniform or frequency based; 
and two configuration for $S_4$ (the longest common subsequence with length 12): a single alphabet and the original 8 distinct alphabets generated by random. The following table shows the occurrence probabilities of $S_4$ in the returned 200 MCS solutions.
\begin{table}[ht]
    \centering
    \begin{tabular}{|c|c|c|}
    \hline
         & uniform weights & frequency-based weights\\\hline
      single alphabet $S_4$   & 0\% & 5\% \\\hline
      8-alphabet $S_4$ & 28\% & 27\% \\\hline 
    \end{tabular}
    \caption{Occurance probabilities for $S_4$, the longest common subsequence by design}
    \label{tab:my_label}
\end{table}
It is clear when $S_4$ is made of all identical characters (i.e, alphabet of size 1), there is a significant drop in the probability of locating the LCS. Nonetheless, random character selection using frequency-weighting performs a lot better. The uniform weights
fails to discover $S_4$, and the longest returned MCS has a length 9. This is because in the case of uniform weights,
the unique alphabets in the long LCS is one of the many to be selected at random with no frequency weighting
and this character is shared by many other MCSs.

We also observe that time to run $RandomMCS$ 200 times is about 150sec for $L=1000$ in our experiment, which is about 50 times for $L=4$ and $1000$ runs (recall the latter instance took about 3sec). This is in agreement with our theoretical analysis of $RandomMCS$ which shows a time complexity linear in $L$.

Our empirical results indicate that LCS typically has a non-negligible occurrence probability among all solutions of MCS and thus will very likely be found by running $RandomLCS$ repeatedly.
However, the performance depends on the nature of LCS and how random search is carried out in the algorithm.

 



\section{Applications to Auto Machine Learning}
In this section, we illustrate how methods we  developed for finding MCSs can be applied to string pre-processing.
Developing automated methods for data pre-processing is an important topic in automated machine learning, or {\em AutoML}, where the objective is to automate the end-to-end process of applying machine learning to real-world problems \cite{feurer2015efficient}. 
We demonstrate how our method can be used to develop a good understanding of string columns in tabular data, and extract important features for downstream machine learning tasks.


\subsection{Data Understanding}
Tabular data is a common form of data representation. It is organized by rows and columns where rows represent individual records and columns are the associated attributes. For large data tables with many rows and columns, it is difficult to obtain a good understanding of the data content without laborious  
manual examination. For columns with string values, we can apply our methods to understand the patterns that
are common across all column values and extract important information or features for downstream machine learning. 

The dataset we use for demonstration contains broadband home router data records of customers from a  network carrier during a 30-day period.  It consists of $27$ columns and $238330$ rows, where columns are device ID and type, associated network node and type, the customer information, and time series of several KPIs. Among the 27 columns, there are 8 columns are either strings or DateTime. For each of these columns, we apply our algorithm to uncover the longest common subsequence  
from 100 runs of $RandomMCS$ algorithm. The resulting patterns are shown
in Table \ref{tab:bt}, where 
we post-processed these common subsequences and represented them in the form of regular expressions where $*$ (the asteroid sign) indicates any number of characters. As a result, 
the contents in the string columns become much more apparent with this information.

\begin{table}[ht]
    \centering
    \begin{tabular}{l | l }
       Colname & Pattern \\\hline\hline
       network.type  & 2*CN* \\\hline  software.version & * \\\hline
       day & 2015-12-*  \\\hline
       customer.attr1 & * \\\hline pop.location & POP-* \\\hline
       linecard.id  & 2*CN*--*--* \\\hline
       sid & BB* \\\hline
        device.id & *0*-Home Hub *0 Type  *-+*+* \\\hline
    \end{tabular}
    \caption{String pattern discovered using our algorithm}
    \label{tab:bt}
\end{table}


\subsection{Feature Extraction}
We can often use the extracted column string patterns in tabular data to engineer new features. 

It is clear that from Table \ref{tab:bt} that some columns have a clear pattern while others do not. For example, both {\em software.version} and {\em customer.attr1} do not have a common pattern. On the other hand, the column of {\em device.id} shows a clear pattern where it can be represented by the string join of 6 sub-fields, each is a combination
of some common characteris shared across the values and a varying substring indicated by asteroid ($*$).
These subfields can be extracted to represent possibly more informative features for characterizing the device.id. This feature extraction step can be automated once patterns are found and the extracted features can be used for downstream machine learning. In fact, our methods can also be applied to auto-detect field separators from an ASCII file and then extract the columns.

%% file: conclusion.tex
\section{Conclusion and Future Work}
 In this paper, we develop a randomized algorithm, referred to as {\em Random-MCS} for finding the maximal common subsequence ($MCS$) of multiple strings. 
We show the complexity of our algorithm is linear in the number of strings $L$.
Furthermore, we demonstrate via both theoretical and experimental studies that the longest subsequence from multiple runs of {\em Random-MCS} often yields a solution to $LCS$. As for future work, we want to improve the  probability bound for a single MCS solution and extend our algorithm to the case when the set of strings is polluted with dirty data. 
\par

%% file: append.tex
\appendix
\section{Extension of Algorithm 1 in Sakai (2019) \cite{sakai2019maximal} to the case of $L$ strings}
The $\mathcal{I}^\prec(A,c,i)$ denotes the least index such that $c$ does not appear in $A(\mathcal{I}^\prec(A,c,i),i]$ and $\mathcal{I}^\succ(A,c,i)$ denotes the greatest index such that $c$ does not appear in $A(i,\mathcal{I}^\succ(A,c,i)]$. The idea is to cut the strings into segments backward and determine the $MCS$ forward. The index sets $idxP$ and $idxR$ mean the previous indices and rear indices. For example, the indices $idxP[j]$ and $idxR[j]$ determine a segment of the $j$ string. So $idxP$ and $idxR$ cut a segment from every string. The Algorithm 4 $Common$ will return $-1$ if no common character exists in all the $L$ segments and return $c$ and $j$ if the common char $c$ appears in the $j$ string first.
\begin{algorithm}
\caption{OneMCS returns a single solution for MCS}
\label{alg:onemcs}
\begin{flushleft}
\textbf{Input}: A List of String $List$ \leavevmode \\
\textbf{Output}: {Single MCS $W$}
\end{flushleft}
\begin{algorithmic}[1]
\STATE Initialize $ W = [\wedge,\text{\textdollar}], W_p \text{\ and\ } W_r$ are vectors with length $L$, $k = 0$.
\STATE \hspace{0in} \textbf{for} $i$ in $1:n$
	\STATE \hspace{0.15in} $\hat{W}[i] = \{0,|A_i|-2\}$ \leavevmode \\
\STATE \hspace{0in} \textbf{while} {k < |$W$|-1}
\STATE \hspace{0.15in}{\textbf{for} $i$ in $1:n$}
\STATE \hspace{0.3in} {$W_p[i] = \hat{W}[i][k], W_r[i] = \hat{W}[i][k+1]$}
\STATE \hspace{0.15in} \textbf{while} {$common(\mathcal{A},W_p,W_r) == -1$}
\STATE \hspace{0.3in}  \textbf{for} $j$ in 1:|$W_r$| \leavevmode \\
\STATE \hspace{0.45in} {$\hat{W}[j][k+1] = W_r[j]-1, W_r[j] = \hat{W}[j][k+1]$}
\STATE \hspace{0.15in} {first = 0} \leavevmode \\
\STATE \hspace{0.15in} first = 1 if $\exists j \ \mbox{such that } W_r[j] == W_p[j]$
\STATE \hspace{0.15in} \textbf{if} first == 1
\STATE \hspace{0.3in} \textbf{for} $j$ in $1:n$
\STATE \hspace{0.45in}$\hat{W}[j][k+1]=\mathcal{I}^\succ(\hat{W}[j],W[k+1],W_p[j]+1)$
\STATE \hspace{0.3in} $k = k+1$
\STATE \hspace{0.15in} \textbf{else}
\STATE \hspace{0.3in} $idx,c = common(\mathcal{A},W_p,W_r)$
\STATE \hspace{0.3in} $W$=$W[1,k]\oplus c \oplus W[k+1,|W|]$
\STATE \hspace{0.3in} \textbf{for} $j$ in $1:n$
\STATE \hspace{0.45in} \textbf{if} $j == idx$
\STATE \hspace{0.6in} {$\hat{W}[j]=\hat{W}[j][:k]\oplus W_r[j]-1\oplus \hat{W}[j][k+1,:]$}
\STATE \hspace{0.45in} \textbf{else} 
\STATE \hspace{0.6in} {$\hat{W}[j] = \hat{W}[j][:k] \oplus \mathcal{I}^\prec   (\hat{W}[j],c,W_r[j])-1) \oplus \hat{W}[j][k+1,:]$}
\STATE \hspace{0in}\textbf{return} $W$
\end{algorithmic}
\end{algorithm}

The Algorithm 3 $OneMCS$ finds a specific $MCS$ for $L$ strings in the complexity $O(nL\log (n))$. Inspired by \cite{sakai2019maximal}, we extend the algorithm from 2 strings to $L$ strings. If it is hard to understand the algorithm $OneMCS$, please read \cite{sakai2019maximal} first.

\begin{algorithm}
\caption{Common returns the common character and its index}
\label{alg:findcommon}
\begin{flushleft}
\textbf{Input}: A List of String $\mathcal{A}=\{A_1,...,A_L\}$, Previous Index $idxP$, Rear Index $idxR$ \leavevmode \\
\textbf{Output}: {The list of indices}
\end{flushleft}
\begin{algorithmic}[1]
\STATE \hspace{0in} \textbf{for} $j$ in 1:$L$:
\STATE \hspace{0.15in}\textbf{if} {$idxP[j] >= idxR[j]$} \leavevmode
\STATE \hspace{0.3in} \textbf{return} {$idxP[j],idxR[j]$}
\STATE \hspace{0in} \textbf{for} {$j$ in 1:$L$}
\STATE \hspace{0.15in} $c = A_j[idxR[j]]$
\STATE \hspace{0.15in} \textbf{for} {$i$ in 1:$L$}
\STATE \hspace{0.3in}  \textbf{if} {$i == j$}
\STATE \hspace{0.45in} \textbf{continue}
\STATE \hspace{0.3in}\textbf{if} $\mathcal{I}^\prec (A_i,c,idxR[i])<=idxP[i]$
 \STATE \hspace{0.45in} \textbf{break}
\STATE \hspace{0.3in}\textbf{if} $i == L-1$
\STATE \hspace{0.45in} \textbf{return} $\{j,c\}$
\STATE \hspace{0.3in} \textbf{if} $j == L-1 \text{ and } i == L-2$
\STATE \hspace{0.45in} \textbf{return} $\{j,c\}$
\STATE \hspace{0in}  \textbf{return} -1
\end{algorithmic}
\end{algorithm}